\documentclass[a4paper,11pt]{article}
\usepackage{mathpazo}
\usepackage{multirow}
\usepackage{nicefrac}
\usepackage{setspace}
\usepackage{fullpage}
\usepackage{xcolor}
\definecolor{xlinkcolor}{cmyk}{1,1,0,0}
\usepackage{hyperref}
\hypersetup{colorlinks=true,linkcolor=black,citecolor=xlinkcolor}
\usepackage{color}
\usepackage{amsthm}
\usepackage{graphicx}
\usepackage{amsmath}
\usepackage{amssymb}
\usepackage{bbm}
\usepackage{dsfont}

\usepackage{diagbox}
\usepackage{makecell}
\usepackage{float}

\newtheorem{theorem}{Theorem}[section]

\newtheorem{corollary}[theorem]{Corollary}
\newtheorem{lemma}[theorem]{Lemma}

\newtheorem{claim}[theorem]{Claim}

\newtheorem{fact}[theorem]{Fact}

\usepackage{xspace}

\newcommand{\kmed}{$k$-$\mathsf{median}$\xspace}
\newcommand{\kmean}{$k$-$\mathsf{means}$\xspace}
\newcommand{\minsum}{$k$-$\mathsf{minsum}$\xspace}

\newcommand{\NP}{$\mathsf{NP}$\xspace}

\newcommand{\poly}{\text{poly}}

\newcommand{\Po}{\mathcal{P}}
\newcommand{\C}{\mathcal{C}}

\newcommand{\Z}{\mathbb{Z}}
\newcommand{\R}{\mathbb{R}}

\newcommand{\half}{\left(\vec{\frac{1}{2}}\right)}
\newcommand{\eps}{\varepsilon}

\renewcommand{\tilde}{\widetilde}

\newcommand{\NN}{\mathbb{N}}      % for Integers

\usepackage{enumerate}
\newcommand{\PTAS}{$\mathsf{PTAS}$\xspace}

\newcommand{\APX}{$\mathsf{APX}$\xspace}

\newcommand{\lam}{\lambda}
\newcommand{\bS}{\overline{S}}
\newcommand{\mkc}{Max $k$-Coverage\xspace}

\newcommand{\calS}{\mathcal{S}}

\usepackage{footnote}
\makesavenoteenv{tabular}
\makeatletter
\newcommand\footnoteref[1]{\protected@xdef\@thefnmark{\ref{#1}}\@footnotemark}
\makeatother

\newcommand{\sfn}{\mathfrak{n}}
\newcommand{\sfr}{\mathfrak{r}}
\newcommand{\na}{\mathfrak{n}_{\alpha}}
\newcommand{\ra}{\mathfrak{r}_{\alpha}}

\newcommand{\minsumvalue}{1.415}

\newcommand{\OPT}{\mathsf{OPT}}
\newcommand{\setcover}{Set Cover\xspace}

\makeatletter
\renewcommand\paragraph{\@startsection{paragraph}{4}{\z@}%
                                      {\parskip}%{3.25ex \@plus1ex \@minus.2ex}%
                                      {-1em}%
                                      {\normalfont\normalsize\bfseries}}
\makeatother

\title{\vspace{-1.5cm}\textbf{
On Approximability of Clustering Problems\\ Without Candidate Centers}}
\date{}
\author{Vincent Cohen-Addad\thanks{
   Google Research, Switzerland.
   \texttt{vcohenad@gmail.com}.
  }
\and 
Karthik C.\ S.\footnote{New York University, USA.
   \texttt{karthik0112358@gmail.com}. 
   }
\and 
Euiwoong Lee\footnote{University of Michigan, USA.
   \texttt{euiwoong@umich.edu}. 
   }
}
\begin{document} 
\maketitle  \vspace{-0.5cm} \thispagestyle{empty} 

\begin{abstract}
  The \kmean objective is arguably the most widely-used cost function for modeling clustering tasks in a metric space.
  In practice and historically, \kmean is thought of in a \emph{continuous} setting, namely where the
  centers can be located anywhere in the metric space. For example, the popular Lloyd's heuristic locates a
  center at the mean of each cluster.\vspace{0.15cm}

  Despite persistent efforts on understanding the approximability of \kmean, and other classic clustering problems such
  as \kmed and \minsum, our knowledge of the hardness of approximation factors of these problems remains quite poor.
  In this paper, we significantly improve upon the hardness of approximation factors known in the literature
  for these objectives. 
  We show that if the input lies in a general metric space, it is \NP-hard to approximate:
  \begin{itemize}
  \item Continuous \kmed to a factor of $2- o(1)$; this improves upon the previous inapproximability factor
    of 1.36 shown by Guha and Khuller (J.\ Algorithms '99).  
  \item Continuous \kmean to a factor of $4- o(1)$; this improves upon the previous inapproximability factor
    of 2.10 shown by Guha and Khuller (J.\ Algorithms '99). 
  \item \minsum to a factor of $1.415$; this improves upon the \APX-hardness shown by Guruswami and Indyk (SODA '03).
  \end{itemize}

Our results shed new and perhaps counter-intuitive light on the differences between clustering problems in the continuous setting versus the discrete setting (where the candidate centers are given as part of the input). 
%\vnote{Possible suggestion if you want to talk about the positive results}
%We also show that our reduction is optimal in several ways: the lower bound cannot be improved even for constant values of $k$ and 
%cannot be extended to $\ell_\infty$ metrics of dimension $o(n)$. 

\end{abstract}

\clearpage \setcounter{page}{1}

\section{Introduction}

Given a set of points in a metric space, a clustering is a partition of the points such
that points in the same part are close to each other.
This makes clustering a basic, crucial computational problem for a variety of applications, ranging from 
unsupervised learning, to information retrieval, and even arching over bioinformatics.
The most popular clustering problem (in metric spaces) is arguably the \kmean
problem: Given a set of points $P$ in a metric space, the \kmean problem
asks to identify a set of $k$ \emph{representatives}, called \emph{centers}, such that the sum of
the squared distances from each point to its closest center is minimized (for the \kmed problem, the goal
is to minimize the sum of distances, not squared) -- see Section~\ref{sec:prelims} for formal definitions.
Finding efficient algorithms that produce good solutions with respect to the \kmean or \kmed objectives has been
a major challenge over the last 40 years.
% Notably, Lloyd's algorithm~\cite{} that dates back to the 70s
% remains a method of choice for finding empirically good solutions. 

From a theoretical standpoint, the picture is rather frustrating: the hardness of approximation for \kmean and
\kmed remain quite far from the approximation that the best known efficient algorithms achieve.
In general metrics, the \kmed and \kmean problems are known to be hard to approximate within a factor of
1.73 and  3.94 respectively~\cite{GuK99}, whereas the best known approximation algorithms achieve
an approximation guarantee of 2.67 and 9 respectively~\cite{BPRST15,ANSW16}. 

The \kmed and \kmean problems come in two flavours: \emph{continuous}, where the set of centers can be picked
arbitrarily in the metric; and \emph{discrete}, where the centers have to be picked from a specific set given
as input. While most of the approximation algorithms known focus on the discrete case, algorithms in practice
(such as e.g. Lloyd method) often leverage the freedom on the location of the centers to get empirically good
performances. In practice, the continuous case is arguably more relevant: when looking for a representative
of a set of points, we would like to find the \emph{best} one and not constraint ourself to some specific set.
In fact, for several metrics such as edit distance, the problem of computing a ``good representative''
of a set of arbitrary strings (i.e.: a string whose sum of distances to the other strings is minimized) is a
well-strudied problem in itself.

At a first glance, it appears that the continuous case is computationally easier than the discrete case, as it allows the algorithm designer not to be forced to pick from the input set of candidate centers. 
In Euclidean space, an important result of Matousek~\cite{Matousek00} shows that an $\alpha$-approximation
algorithm for the discrete case of \kmean can be used to obtain a $(1+\eps)\cdot \alpha$-approximation to
the continuous case of \kmean under the $\ell_2$ distance.
This suggests that the continuous case is somewhat \emph{easier} than the discrete case in the Euclidean metric\footnote{Note that we know non-trivial inapproximability results  for Euclidean \kmean and \kmed\cite{ACKS15,LSW17,CK19}.}. Moreover, the 20-year
old hardness results of Guha and Khuller~\cite{GuK99} of $1+2/e$ and $1+8/e$ for \kmed and \kmean respectively
only apply to the discrete case and the only known bounds for the continuous setting derived from their approach
are $1+1/e \approx 1.36$ and $1+3/e \approx 2.10$ for \kmed and \kmean respectively. 
We thus ask: 

\begin{center}
\textit{Can we approximate continuous \kmean (resp.\ continuous \kmed)\\ to a factor less than $1+\nicefrac{8}{e}$ (resp.\ $1+\nicefrac{2}{e}$) in polynomial time? }
\end{center}

Another classic clustering objective in the \minsum problem. 
%  \knote{maybe we want to provide a sentence saying this objective enforces that clusters are somewhat balanced?}.
Given a set of points in a metric space, the \minsum problem asks for a partition of the points to $k$ parts that minimizes
the sum of the pairwise distances between points in the same part of the partition (see Section~\ref{sec:prelims} for formal definition).
Compared to \kmean and \kmed, the fact that the objective function sums over a quadratic number of distances within each cluster favors {\em balanced} clustering where clusters are of similar sizes. 
This fundamental clustering problem introduced in the 70s by Sahni and Gonzalez~\cite{sahni1976p},
together with the capacitated \kmed problem, is one of the problems for which designing an $O(1)$-approximation algorithm
or showing that none exists, for general metric case remains an important open problem. 

The \minsum problem has received a large amount of attention over the
years~\cite{guttmann1998approximation,schulman2000clustering,indyk2000high,VegaKKR03, czumaj2004sublinear, czumaj2010small},
but the current understanding of \minsum is worse than that of \kmed and \kmean: 
while no better than $O(\log n)$-approximation
is known in polynomial time~\cite{bartal2001approximating, Behsaz2019}, the best known hardness of
approximation factor is $(1+\eps)$, due to Guruswami and Indyk~\cite{GI03}, for some small implicit constant $\eps > 0$.
Getting better hardness of approximation for the \minsum remains an important open problem.
Arguably, the intrinsic \emph{continuous} nature of the problem -- the fact that the hardness must be directly encoded
into the locations of the points -- has been one of the most important roadblock for the problem.

\begin{center}
\textit{Can we show hardness of approximation result for \minsum\\ for any explicit, non-negligible constant greater than 1? }
\end{center}

\paragraph{Technical Barriers.}
A well-known framework to obtain hardness of approximation results in the general metric for clustering objectives is through a straightforward reduction from the Max $k$-Coverage or the Set
Cover problem.  
Given an instance of \mkc that consists of a collection of its subsets of some universe, 
we create a 'point' for each element of the universe and a 'candidate center', namely a location
where it is possible to place a center, for each set. Then, we define the distance between a point (corresponding to an element of the universe) and a candidate center (corresponding to a set) to be 1 if the set contains the element and 3 otherwise. 
This reduction due to Guha and Khuller~\cite{GuK99} yields lower bounds of $1+2/e$ and $1+8/e$ for the \kmed and \kmean problems, respectively, in general discrete
metric spaces.

The reduction of Guha and Khuller~\cite{GuK99} for \kmed in general metrics does not even rule out \PTAS for \minsum,
mainly due to the fact that even in one cluster, the objective function sums over all pairs of points whose edges may come from different sets. 
To bypass this issue, the only known \APX-hardness~\cite{GI03} starts from a very restricted set system where every set has $3$ elements and only rules out $(1+\eps)$ factor approximation algorithms for some implicit constant $\eps > 0$. However, reductions form bounded degree set systems are highly restrictive and one cannot typically hope to prove inapproximability for factors $1+\alpha$, for non-negligible $\alpha$. 

One may thus wonder if  there are other structured set systems which could be the  right starting point for proving hardness of approximation results for \minsum. In fact one may further wonder if the hard instances of clustering problems as a whole are completely captured by hard instances of various kinds of set systems or maybe there are other mathematical objects which might be more appropriate to prove improved inapproxiability results for  certain clustering problems.

\subsection{Our Results}
\label{sec:results}

%\knote{Make it explicit: minsum is tech contribution and kmed/kmean is conceptual contribution. }
%\enote{I think they are both conceptual contribution; techniques for minsum are nothing new but the focus should be on hardness of max coverage on structured set systems thing. }

The main contributions of this paper are conceptual. First, we develop an approach to provide the first \emph{explicit} constant inapproximability ratio for the \minsum problem. 
En route to proving the inapproximability of \minsum, 
we also prove that the $(1 - 1/e)$-hardness of approximation for \mkc holds, even for set systems of bounded {\em VC dimension} --- an important notion in computational geometry and machine learning. 
We believe that further study on approximability of \mkc restricted to set systems with additional combinatorial and geometric structures will produce not only interesting results on their own but also have wide applications. 
We discuss the details about the result and the technique further in Section~\ref{sec:introminsum}.

%\knote{Emphasize a bit more the technical novelty.}\enote{I am shy about bragging it..}

%As the VC dimension is an important quantity studied in the intersection of combinatorics and geometry, 
%we believe that this result justifies the effectiveness of studying such objects through the lens of computation, suggested in the last paragraph.
%This result may be interpreted as studying an intermediate mathematical quantity between combinatorics and geometry, which we  suggested in the previous paragraph would act like a bridge to connect inapproximabilty results in graph theory (or more generally discrete mathematics) to the hardness of approximating geometric optimization problems. 

Our second contribution is an  insight for proving hardness of approximation results for \emph{continuous} versions of \kmean and \kmed  in general metrics\footnote{We write the result in this paper for the $\ell_\infty$-metric, but the reader should note that there is a Fr\'echet embedding from any discrete metric to the $\ell_\infty$-metric in high dimensions.}.  In particular, instead of starting the reduction from set-cover-type problems we start from coloring problems and yield a surprising result that the complexity of the discrete and continuous versions are significantly different, but in the counter-intuitive direction --- the
 continous version of the problem is harder to approximate than the discrete version! This is elaborated further in Section~\ref{sec:introinf}. 
%\knote{Emphasize a bit more about doing something completely different from Guha and Khuller and move some of the discussion above to 1.1.2.}

\begin{table}[H]
\begin{center}{\renewcommand{\arraystretch}{1.45}
\resizebox{0.9\linewidth}{!}{
\begin{tabular}{||c|c|c|c|c||}
    \cline{1-4}
    Objective  & \thead{Continuous\\ \kmean}& \thead{Continuous\\ \kmed}&\minsum & \multicolumn{1}{c}{}  \\ [0.5ex]\cline{1-4}\hline
  \multirow{2}{*}{Hardness}  &  2.10  \cite{GuK99} &  1.36 \cite{GuK99}&\APX-Hard \cite{GI03}&{\footnotesize Previous}\\ 
     &  \textbf{4}& \textbf{2}&$\mathbf{\minsumvalue} $&{\footnotesize This paper}\\ \hline
  {Algorithms}  & 36  \cite{kanungo2002local} &5.3 \cite{BPRST15}  &$O(\log n)$ \cite{Behsaz2019} & \multicolumn{1}{c}{} \\ \cline{1-4}
\end{tabular}
}
}
\end{center}
\vspace{-0.1cm}
\caption{{State-of-the-art approximability results for clustering objectives without candidate centers in general metric. 
The algorithmic results for \kmean and \kmed, though not explicitly stated in the literature, can be obtained by considering data points as candidate centers (which loses a factor $4$ and $2$ for \kmean and \kmed respectively) and 
running the algorithms for the discrete problems cited in the references. 
}
}\label{table}
\end{table}

\subsubsection{Inapproximability Results for \minsum}\label{sec:introminsum}
%% \vnote{I like this part which is crisper. Maybe we could place that before the continuous business.}
We state our results on the \minsum problem.
\begin{theorem}[\minsum in $\ell_{\infty}$-metric]\label{thm:introminsum}
\begin{sloppypar}Given $n$ points  in  $O(\log n)$ dimensional $\ell_{\infty}$-metric space it is \NP-hard (under randomized reductions) to distinguish between the following two cases:
\begin{itemize}
\item \textbf{\emph{Completeness}}:  
The \minsum objective is at most 1.
\item \textbf{\emph{Soundness}}: The \minsum objective  is at least $\minsumvalue$.
\end{itemize}
 \end{sloppypar}
\end{theorem}

In order to prove Theorem~\ref{thm:introminsum}, we prove hardness of \mkc in a specialized set system. 
Given an instance $(U, E, k)$ for \mkc where $U$ is the universe and $E$ is a collection of subsets, 
let the {\em girth} of the set system $(U, E)$ to be the girth of the incidence bipartite graph; 
the vertex set is $U \cup E$ and there is an edge $(u, S) \in U \times E$ if and only if $u \in S$.
% \knote{Maybe we need to have an explicit theorem statement for above statement,}
When the girth of a set system is strictly greater than $4$, then no two sets intersect in more than a single element,
so the {\em VC dimension} of the set system is also at most $2$.
Set systems with bounded VC dimensions are known to admit qualitatively better algorithms
such as $O(\log \OPT)$-approximation algorithm for \setcover~\cite{bronnimann1995almost} 
and an FPT-approximation scheme for \mkc~\cite{badanidiyuru2012approximating},
which cannot exist for general set systems~\cite{KLM18, M20}. 

We prove a hardness result showing that, for polynomial time approximation for \mkc, having a bounded VC dimension
(even a super-constant girth) does not help.

\begin{theorem} [Informal statement of Theorem~\ref{thm:mkc}]\label{thm:intromkc}
\begin{sloppypar}
For any $\eps > 0$, it is \NP-hard (under randomized reductions) to approximate \mkc within a factor of $(1 - 1/e + \eps)$ 
even when the set system has girth $\Omega(\log n / \log \log n)$ and maximum degree $O_{\eps}(1)$.
\end{sloppypar}
\end{theorem}

%In particular, we show that for an $\eps>0$, \mkc remains \NP-hard to approximate to within a factor of $(1 - 1/e + \eps)$ even  the incidence graph of \mkc instance has girth $\Omega(\log n / \log \log n)$ %and maximum degree $O_{\eps}(1)$ (see Theorem~\ref{thm:mkc}). 

\begin{sloppypar}The above result is proved by ``lifting'' Feige's optimal hard instances of \mkc~\cite{F98}.
Given a hard instance of \mkc without any girth guarantee, 
we take the dual set system to view it as a hypergraph vertex coverage problem.
For each vertex, we create a {\em cloud} of many vertices, and for each hyperedge, 
we create many random copies where each copy contains a random vertex in each cloud.

Intuitively, putting too many hyperedges will result in many intersections between hyperedges, which may create a short cycle. 
On the other hand, putting too few hyperedges will make the new instance significantly different from the original instance, 
possibly allowing a small hitting set that does not reveal the hitting set in the original hypergraph.
By appropriately choosing the size of cloud and the number of hyperedges and carefully analyzing the probabilities for both bad events,
it can be shown that the hardness is almost preserved while the girth becomes large. %\knote{We can provide more intuition and provide more details here. }
\end{sloppypar}

Given the hardness of \mkc with large girth, the reduction to \minsum is simple; 
given a set system for \mkc, the instance for \minsum is given by the graphic metric
where each vertex corresponds to an element and two vertices are connected if the corresponding elemtns are contained in the same set.
If the set system can be partitioned into $k$ sets in the system, the graph can be partitioned into $k$ cliques, so
every pair of vertices in the same cluster are at distance 1 from each other. 
To analyze the soundness, even though edges within one cluster may come from different sets, 
the girth $\Omega(\log n / \log \log n)$ is larger than the average cluster size (which is still bounded by $O_{\eps}(1)$), 
so we can argue that most clusters, roughly correspond to only one set of \mkc. 

\subsubsection{Inapproximability Results for Continuous \kmean and \kmed in General Metric Space}\label{sec:introinf}
Finally, we state below the inapproximability of \kmed and \kmean in the continuous case for the $\ell_\infty$-metric, whose factors are even higher than that of \cite{GuK99} for \kmed and \kmean in the discrete case\footnote{By applying Fr\'echet embedding, we can embed any discrete metric into the $\ell_\infty$-metric, preserving all pairwise distances. } for the $\ell_\infty$-metric. 

\begin{theorem}[Informal statement of Theorems~\ref{thm:kmeanellinfhighdim}, \ref{thm:kmedellinfhighdim}, and \ref{thm:2meanmedellinfhighdim}]\label{thm:introellinf}
For every constant $\varepsilon>0$, there exists a constant integer $k$ such that, given $n$ points  in  $\poly(n)$ dimensional
    $\ell_{\infty}$-metric space it is \NP-hard to approximate:
  \begin{itemize}
  \item the \kmean objective to within $ 4-\varepsilon$ factor.
  \item the \kmed objective to within $2-\varepsilon$ factor.
    \end{itemize}
  Moreover, the above statement holds for $k =4$ (and can be further strengthened to hold for $k=2$ by assuming the Unique Game Conjecture).
\end{theorem}

The above result is very surprising as it breaks the more than twenty year old bound of \cite{GuK99}. Furthermore it is believed  that the bound of \cite{GuK99} is indeed tight for the discrete case as there are $1+2/e$ and $1+8/e$ parameterized approximation algorithms for \kmed and \kmean problems respectively  in general metrics \cite{Cohen-AddadG0LL19} (note that this is merely an indication that \cite{GuK99} might be tight for the discrete case and not a formal conclusion).  Therefore this provides \emph{morally} the first separation between the continuous and discrete versions for clustering problems. 

 Further, we show that the bound in Theorem~\ref{thm:introellinf} is tight for a large range of settings.
 First, for any constant $k$, we note that there is a simple  2-approximation algorithm to the
 continuous \kmed problem and a 4-approximation algorithm for the continuous \kmean problem
 in the $\ell_{\infty}$-metric  both running in polynomial time. 
 Second, we show that the hardness result with the same gap cannot hold for much smaller dimensions (see  Corollaries~\ref{cor:below2},~\ref{cor:fptk}~and~\ref{cor:1+2/e}).

The proof of Theorem~\ref{thm:introellinf} follows from a new technique to construct clustering problem inputs; instead of starting from set-cover-type problems  (as in the framework of \cite{GuK99}), we  start our reductions from the hard instances of  $k$-coloring (or equivalently on finding $k$-disjoint independent sets) in graphs due to \cite{KS12}. In other words, instead of starting from covering problems on graphs (like almost all other results in literature) and embedding a pair of vertices sharing an edge as points that are close  and other vertex pairs far away, we start from the complement of cover problems, i.e., the independent set problem and embed  a pair of vertices \emph{not} sharing an edge as points that are close and other vertex pairs far away, leveraging the stronger inapproximability of the independent set problem.  %\knote{Elaborate more slowly.}

\subsection{Organization of the Paper}
The paper is organized as follows. In Section~\ref{sec:prelims}, we introduce some notations that are used throughout the paper.
 In Section~\ref{subsec:mkc}, we prove  our hardness of approximation result for \mkc on instances with large girth (i.e., Theorem~\ref{thm:intromkc}).
  In Section~\ref{sec:minsum}, we prove our hardness of approximation result for \minsum objective in general metrics (i.e., Theorem~\ref{thm:introminsum}). In Section~\ref{sec:cont}, we prove our improved inapproximability results for \kmean and \kmed in general metrics (i.e., Theorem~\ref{thm:introellinf}).

\section{Preliminaries}
\label{sec:prelims}

\paragraph*{Notations.} For any two points $a,b\in\mathbb{R}^d$, the distance between them in the $\ell_{\infty}$-metric is denoted by 
$\|a-b\|_{\infty} = \underset{{i\in[d]}}{\max}\ \{|a_i-b_i|\}$.  Let $e_i$ denote the vector which is 1 on coordinate $i$ and 0 everywhere else. We denote by $\half$, the vector that is $\nicefrac{1}{2}$ on all coordinates.

\paragraph*{Clustering Objectives.} Given two sets of points $P$ and $C$ in a metric space, we define
 the \kmean cost of $P$ for $C$
 to be $\underset{p \in P}{\sum}\left(\underset{c \in C}{\min}\ \left(\text{dist}(p,c)\right)^2\right)$ and
 the \kmed cost to be $\underset{p \in P}{\sum}\left(\underset{c \in C}{\min}\ \text{dist}(p,c)\right)$.
 Given a set of points $P$  in a metric space and partition $\pi$ of $P$ into $P_1\dot\cup P_2\dot\cup\cdots \dot\cup P_k$, we define
 the \minsum cost of $P$ for $\pi$
 to be $\underset{i\in [k] }{\sum}\left(\underset{p,q \in P_i}{\sum}\ \text{dist}(p,q)\right)$.
 Given a set of points $P$, the \kmean/\kmed (resp.\ \minsum)
 objective is the minimum over all $C$ (resp.\ $\pi$) of cardinality $k$ of the \kmean/\kmed (resp.\ \minsum)
 cost of $P$ for $C$ (resp.\ $\pi$). Given a point $p \in P$, the contribution to the
 \kmean (resp.\ \kmed) cost of $p$ is $\underset{c \in C}{\min}
 \left(\text{dist}(p,c)\right)^2$ (resp.\ $\underset{c \in C}{\min}\ \text{dist}(p,c)$).

\section{Hardness of \mkc with large girth} 
\label{subsec:mkc}

In this section, we prove the following hardness of \mkc with large girth and bounded degree and then use the hardness result to prove Theorem~\ref{thm:minsum} for \minsum clustering in the next section. 
Like \kmed   \cite{GuK99}, the result is based on hardness of \mkc;
given an instance $(U, E, k)$ of \mkc, we output the corresponding instance of \minsum consisting a graph $G = (U, E')$ where
$v, u \in U$ have an edge if and only if there exists $S \in E$ that contains both $u$ and $v$. 
However, unlike \kmed, just the objective function value of \mkc does not suffice to prove results for \minsum.
For example, consider an instance of \mkc where typical sets are large, 
but we add a set of size two for each pair of elements. These sets of size two are small so that it will not affect the \mkc objective function,
but the outcome of the reduction will be a complete graph! 
Therefore, we need to start from hardness of \mkc in a specialized set system.

The proof starts from the standard \mkc hardness result of Feige~\cite{F98} that has no guarantee on girth.
Considering the dual set system has a hypergraph, we put many copies of each vertex and many random copies of each hyperedge.
This idea was previously used in subgraph hitting sets and constraint satisfaction problems~\cite{guruswami2015inapproximability, ghosh2017weak}.

\begin{theorem}
\begin{sloppypar}
For any $\eps > 0$, given an instance $(U, E, k)$ is \mkc 
where the incidence graph has girth $\Omega(\log n / \log \log n)$ and maximum degree $O_{\eps}(1)$, 
it is \NP-hard (under randomized reductions) to distinguish between the following two cases:
\begin{itemize}
\item \textbf{\emph{Completeness}}:  
There exists $k$ sets that cover $E$. 

\item \textbf{\emph{Soundness}}: 
Any $k$ sets cover at most an $(1 - 1/e + \eps)$ fraction of $E$. 
\end{itemize}
\end{sloppypar}
\label{thm:mkc}
\end{theorem}
\begin{proof}
We consider the dual set system of the hard instance of \mkc given by Feige~\cite{F98} as a regular $r$-uniform hypergraph $H_0 = (V_0, E_0)$, which has $n$ vertices, $m$ hyperedges, and degree $d$ (so that $nd = mr$). In the completeness case, there is a set $S^* \subseteq V_0$, $|S^*| = k = n / r = m / d$ that intersects every hyperedge $e \in E_0$. In the soundness case, any set $|S| \leq k$ hits at most $(1 - 1/e +\eps)$-fraction of hyperedges. Feige's reduction also ensures that this hardness can be achieved with $r$ and $d$ being constants (depending on $\eps$). 

The new hypergraph $H = (V, E)$ is the following. Let $\ell$ and $B$ be numbers determined later (they will be both $\Theta(n)$).
\begin{itemize}
\item $V = V_0 \times [B]$. 
\item For each $e \in E_0$,
\begin{itemize}
\item For each $v \in e$, sample $(j_{1, v}, \dots, j_{\ell, v}) \in [B]^{\ell}$ uniformly from the set of $\ell$-tuples where every number in $[B]$ appears the same number of times (we will ensure $B$ divides $\ell$). 
\item For each $i \in [\ell]$, add $\{ (v, j_{i, v}) \}_{v \in e}$ to $E$. 
\end{itemize}
\item For each simple cycle of the incidence bipartite graph of length at most $t$ (which will be fixed later), delete an arbitrary hyperedge in it. 
\end{itemize}
Then $|V| = |V_0| \cdot B$, $|E| \leq |E_0| \cdot \ell$. Note that the girth is at least $t$, and the maximum degree is at most $d \cdot \Theta(\ell / B) = O(1)$. 

\paragraph{Girth control.}
We bound how many hyperedges we deleted in the last step of the construction. 
Consider the incidence bipartite graph of the hypergraph; hyperedge vertices are (a subset of) $E_0 \times [\ell]$ and element vertices are $V_0 \times [B]$. 
Fix a $2t$-tuple 
\[
((v_1, p_1), (e_1, q_1), (v_2, p_2), (e_2, q_2), \dots, (v_t, p_t), (e_t, q_t)), 
\]
where all vertices are different and $v_i, v_{i + 1} \in e_i$ (and $v_1 \in e_t$). 
We have $n$ choices for $v_1$, and after that $d$ choices for each $e_i$ and $r$ choices for each $v_i$, so the number of such tuples is upper bounded by 
\[
n \cdot (dr)^t \cdot (B \ell)^t. 
\]
For each possible edge in the tuple (say $((v_i, p_i), (e_i, q_i))$), the probability that it appears is the probability that 
$j_{q_i, v_i} = p_i$ in the above sampling procedure for $e_i$. 
Since $j_{q_i, v_i}$ draws from $B$ numbers and we will take $t = o(\log n) \ll B$, this probability, conditioned on existence of an arbitrary set of edges in the tuple, is at most $2/B$. 
%Given a fixed tuple, the probability that it becomes a cycle is at most $1/B^{2t}$, because each potential edge appears independently or in a negatively correlated way with probability at most $1/B$. 
%(Because all $(e_i, q_i)$'s are different and independent. The only dependence between cycle edges can happen when $v_i = v_{i + 1}$, and in this case the probability of being a cycle is $0$.)
So the expected number of cycles is at most 
\[
n \cdot (dr)^t \cdot (B \ell)^t \cdot (2/B)^{2t} = n \cdot (dr)^t (4 \ell / B)^t. 
\]
We will take $\ell = a B$ for some constant $a$ depending on $r$ and $\eps$. Let $B = n$. 
Using Markov's inequality, with probability at least $3/4$, the number of hyperedges we deleted is at most 
\[
4n \cdot (dr)^t (4 \ell / B)^t = 
4n \cdot (4adr)^t = o(m\ell). 
\]
as long as $t = o(\log n)$. Fix $t = \Omega(\log n / \log \log n)$.
We can ensure that the girth is at least $t$ with losing only $o(1)$ fraction of hyperedges. 

\paragraph{Completeness.}
If $S \subseteq V_0$ is a feasible solution for the \mkc instance (i.e., $S$ intersects every $e \in H_0$), then $S \times [B]$ is a feasible solution for the new instance.

\paragraph{Soundness.}
Fix a hyperedge $e \in E_0$. For simplicity let us assume $e = (v_1, \dots, v_r)$. 
Fix $C_1, \dots, C_r \subseteq [B]$, and let $\alpha_i := |C_i| / B$. 
We want to show that out of $\ell$ hyperedges in the new instance coming from $e$, approximately $1 - \prod_{i=1}^r (1 - \alpha_i)$ fraction of hyperedges intersect $\cup_{i=1}^r (v_i \times C_i)$.
For one such hyperedge, the probability is exactly $1 - \prod_{i=1}^r (1 - \alpha_i)$.
The $\ell$ hyperedges are not independent, but since the distribution is {\em negatively correlated} (i.e., if one hyperedge intersects $\cup_{i=1}^r (v_i \times C_i)$, other hyperedges are less likely to intersect it.) We can still apply the Chernoff bound so that the probability that the total number is $\eps \ell$ more than the expectation is at most $\exp(- \Theta(\eps^2 \ell))$. Since there are at most $2^{Br}$ choices of $C_1, \dots, C_r$ and we let $\ell = aB$, with probability at most 
\[
2^{Br} \cdot \exp(-\Theta(\eps^2 a B)) \leq \exp(B(r - \Theta(\eps^2 a))),
\]
which is exponentially small in $B$ (thus $n$) if we take $a$ to be a large constant depending on $r$ and $\eps$. 
Union bounding over all $e \in E_0$, we showed that for any $S \subseteq V$ for the new instance with $|S| \leq kB$, if we let $\alpha_v := |S \cap (v \times [B])| / B$ (so that $\sum_v \alpha_v = k$), then the fraction of hyperedges $S$ intersects in the new instance is at most $\eps$ more than the expected fraction of hyperedges hit in the old instance if we {\em round} each $v \in V_0$ independently with probability $\alpha_v$. In the soundness case the latter is at most $(1-1/e+\eps)$, so with high probability the optimal value in the new instance is at most $(1-1/e+2\eps)$. 
\end{proof}

\iffalse
\paragraph{Max Degree.}
In the reduction, if vertex $v \in V_0$ has degree $d$ in the original hypergraph, since we sample $\ell = aB$ hyperedges for each original hyperedge, the expected degree of any $(v, i)$ in the new hypergraph is exactly $ad$. Indeed, we can make sure that each vertex in the new instance is $ad$-regular (before the deletion.) Feige's reduction ensures that the original hypergraph is $d$-regular, and when we sample $\ell = aB$ hyperedges for one original hyperedge, we can make sure that each vertex is contained in exactly $a$ of them. The Chernoff bound will go through because this distribution will be negatively correlated. So even after deletion, if we consider the primal set system, each set has at most $ad = O(1)$ elements.\footnote{Actually there is some more check to do in girth control, but should be fine.}
\fi

To prove hardness of \minsum, we additionally need to prove the in the soundness case, no $\alpha k$ sets cover more than an $(1 - e^{-\alpha})$ fraction of elements for any constant $\alpha > 0$. 
The same construction ensures it.
\begin{corollary}
Theorem~\ref{thm:mkc} holds with the following stronger soundness: For any constant $\alpha > 0$, 
\begin{itemize}
\item Soundness: Any $\alpha k$ sets cover at most an $(1 - e^{-\alpha} + \eps)$ fraction of $E$. 
\end{itemize}
\label{cor:mkc}
\end{corollary}
\begin{proof}
Guha and Khuller~\cite{GuK99} proved that the same soundness for general set systems. 
Their result uses a tight $((1 - \eps) \ln n)$-hardness of Set Cover whose reduction took time $n^{O(\log \log n)}$ at that time, but the running time became polynomial~\cite{dinur2014analytical}.
The proof of Theorem~\ref{thm:mkc} indeed shows that the maximum fraction of elements covered by any $\beta$ fraction of sets in the new set system is at most $\eps$ plus the same quantity in the original set system, so we can transfer this strong hardness for general set systems to set systems of high girth, up to an additive $\eps$ factor. 
\end{proof}

\section{Inapproximability of \minsum in General metric}\label{sec:minsum}

In this section, we use Theorem~\ref{thm:mkc} to prove hardness of \minsum clustering. 
The reduction is simple; given an instance $(U, E, k)$ of \mkc, we output the corresponding instance of \minsum consisting a graph $G = (U, E')$ where $v, u \in U$ have an edge if and only if there exists $S \in E$ that contains both $u$ and $v$. Therefore, if each cluster is a clique of $G$, then each pairwise distance is $1$, and if it is a sparse subgraph of $G$, then the average pairwise distance is approximately at least $2$. 
Using the large girth guarantee in Theorem~\ref{thm:mkc}, we prove that any dense induced subgraph of a certain size must correspond to elements covered by a single set, so that any good solution for \minsum implies a good solution for \mkc. Since the objective function considers all pairwise distances in each cluster, more technical calculations are needed to prove a better inapproximability factor.

\iffalse
\begin{theorem}
For any $\eps > 0$, there exists $r = O_{\eps}(1)$ such that given a graph $G = (V, E)$ with $n = |V|$ and $k \in \NN$ that satisfy $n = (1 - o(1))kr$, it is hard to NP-hard to distinguish
\begin{itemize}
\item YES: $V$ can be partitioned into $k$ cliques of size at most $r$.
\item NO: For any partition $V_1, \dots, V_{k}$ of $V$ satisfying $\max_i |V_i| \leq r$, average density of $V_i$ is at most $(1 - 1/e + \eps)$. 
\end{itemize}
\label{thm:graph}
\end{theorem}
\fi

\begin{theorem}[Restatement of Theorem~\ref{thm:introminsum}]\label{thm:minsum}
\begin{sloppypar}Given $n$ points  in  $O(\log n)$ dimensional $\ell_{\infty}$-metric space it is \NP-hard (under randomized reductions) to distinguish between the following two cases:
\begin{itemize}
\item \textbf{\emph{Completeness}}:  
The \minsum objective is at most $\beta$,
\item \textbf{\emph{Soundness}}: The \minsum objective  is at least $\minsumvalue \cdot \beta$,
\end{itemize}
where $\beta$ is some positive real number depending only on $n$. \end{sloppypar}
\end{theorem}

\begin{proof}
Given an instance $\calS$ of \mkc promised in Theorem~\ref{thm:mkc}, where
the maximum set size $r = O(1)$ and the incidence bipartite graph has max degree $O(1)$ and girth $t = \Omega(\log n / \log \log n)$, 
let $G = (V, E)$ be the graph where $V$ consists of elements, and for each set $S$, we put a clique on its elements. Since the girth of the set system is at least $2$, these cliques are all edge disjoint. 
Note that $n = (1 - o(1))kr$ from Theorem~\ref{thm:mkc}.
The instance for \minsum clustering is the shortest metric on $G$ along with the same $k$. 

Indeed, since our analysis only uses distances $1$ and $2$, we can change all distance greater than $2$ to $2$.
Guruswami and Indyk~\cite{GI03} showed that any $\{1, 2\}$-metric where each point has only $O(1)$ other points at distance $1$ can be embedded to $O(\log n)$-dimensional $\ell_{\infty}$ space, which can be applied to our metric because each vertex in $G$ only has $O(1)$ neighbors. 

\paragraph{Completeness.} In the completeness case of Theorem~\ref{thm:mkc}, we can partition $G$ into $k$ cliques, each of size at most $r$. 
The clustering cost is then at most $k \cdot \binom{r}{2} \leq (1 + o(1))nr/2$.

\paragraph{Soundness.} Fix $V' \subseteq V$ and let $n' := |V'|$. Consider $V'$ as one cluster. We will bound the \minsum cost of $V'$ as one cluster. 
Consider the following cases. 

\begin{enumerate} 
\item $n' \leq t / 2$: Consider the set system $\calS'$ {\em induced by (in the bipartite graph sense) } $V' \cup \{ S : |S \cap V'| \geq 2, S \in \calS \}$.
The corresponding bipartite graph is acyclic, so a forest. 
Let $S'_1, \dots, S'_{m'}$ be the sets of this restricted system, and let $a'_i := |S'_i|$.
Let $r' := \max_i a'_i$. 

We want to upper bound $\sum_i \binom{a'_i}{2}$.
For each tree in the forest, root it at an arbitrary element vertex. For each $S'_i$ we get $\binom{a'_i}{2} = a'_i (a'_i - 1) / 2$. Charge this to its $a'_i - 1$ children, $a'_i / 2$ each. 
Since $a'_i \leq r'$, every element vertex is charged at most $r' / 2$. 
This shows that $\sum_i \binom{a'_i}{2} \leq r' n' / 2$. 
When $r' > n' / 2$, using the fact that all other $a'_i \leq n' - r'$, we have a better bound of 
$(r')^2 / 2 + (n' - r')^2 / 2$. 
Note that $\sum_i \binom{a'_i}{2}$ is exactly the number of edges in the subgraph of $G$ induced by $V'$. 
Therefore, the cost of $V'$ is at least
\begin{align*}
&2 \cdot \binom{n'}{2} - \min(r' n' / 2, \quad (r')^2 / 2 + (n' - r')^2 / 2) \\
=& (1 - o_r(1)) \max\bigg( ((n')^2 - n' r' / 2), \quad (n')^2 / 2 + n' r' - (r')^2 \bigg).
\end{align*}
Here $o_r(1)$ denotes a quantity decreasing to $0$ as $r$ increases. By taking $r$ large enough (but still) constant, we can ignore up to an arbitrarily small additive factor in the final inapproximability ratio. 

\item If $n' > t / 2$. Since (the bipartite graph of) $\calS$ has degree $O(1)$, $G$ also has degree $O(1)$. 
Therefore, if $V' \subseteq V$ has $n' = |V'| \geq t / 2 = \Omega(\log n / \log \log n)$, the induced graph $G_{V'}$ has density at most $o(1)$, so the cost is at least 
$(2 - o(1)) \binom{n'}{2}$. 
\end{enumerate}

Now we compute the \minsum cost for a $k$-clustering. 
Let $V_1, \dots, V_k$ be a partition of $V$ and let $n_i := |V_i|$. 
Let $r_i$ be the largest clique size in $G_{V_i}$ (same as $r'$ in the case (1)).

Suppose that $n_i \geq t/2$ for each $i \in [\ell]$. The total cost from these $\ell$ clusters is at least 
$
(2 - o(1))  \sum_{i = 1}^{\ell} \binom{n_i}{2}.
$
If $\sum_{i = 1}^{\ell} n_i = \Omega(n / \log \log n)$, since $t = \Omega(\log n / \log \log n)$ \allowdisplaybreaks
\[
(2 - o(1)) \sum_{i = 1}^{\ell} \binom{n_i}{2}
\geq (2 - o(1)) \cdot \frac{t - 1}{2} \cdot \sum_{i = 1}^{\ell} n_i = \Omega(n \log n / (\log \log n)^2),
\]
which is superconstant times larger than the cost $(1 + o(1)) nr/2$ in the completeness case. 
Therefore, we can conclude that clusters of size at least $t/2$ cover at most an $o(1)$ fraction of vertices, so up to an $(1 - o(1))$ factor we can assume that every $V_i$ satisfies $n_i \leq t_i / 2$. 
Then the above case 1 is applied for every $V_i$, so the total cost at least (again up to a $(1 - o_r(1))$ factor), 
\begin{equation}
\sum_{i = 1}^{k} \max(n_i^2  - r_i n_i  / 2, \,\, n_i^2/2 + n_ir_i - r_i^2). 
\label{eq:minsumf}
\end{equation}

Let $f(n, r) := \max(f_1(n, r), f_2(n, r))$, with
$f_1(n,r) := n^2  - r n  / 2$ and $f_2(n, r) := n^2/2 + nr - r^2$. 
Note that $f_1(n, r) = f_2(n,r)$ when $r = n/2$. 
For any fixed $n_i$, it can be checked that 
$f(n_i, r_i)$ is decreasing in $r_i$. 
Therefore,~\eqref{eq:minsumf} is minimized when $r_i$'s are as large as possible. 
So we can apply Corollary~\ref{cor:mkc} and assume that the worst case for~\eqref{eq:minsumf} happens (up to an $(1+o(1))$ factor) when $r_i = r \cdot e^{-i/k}$. 

For the sake of exposition, we let $\sfn_{\alpha} := n_{\alpha k} / r$, $\sfr_{\alpha} := r_{\alpha k} / r$ for $\alpha \in [0, 1]$. 
So~\eqref{eq:minsumf} becomes
\begin{equation}
\sum_{i = 1}^{k} f(n_i, r_i) = 
k r^2 \sum_{i = 1}^{k} \bigg( f(\sfn_{i/k}, \sfr_{i/k}) \cdot (1/k) \bigg)
= (1 \pm o(1)) nr \cdot \int_{\alpha =0}^1 f(\sfn_{\alpha}, \sfr_{\alpha}) d\alpha,
\label{eq:minsumff}
\end{equation}
where we use linear interpolation to extend $f$ to all $[0, 1]$. Given $\sfr(\alpha) = e^{-\alpha}$, we find the best $(\sfn_{\alpha})_{\alpha \in [0, 1]}$ to minimize~\eqref{eq:minsumff}.
There are three requirements for $(\sfn_{\alpha})$. 
\begin{enumerate}
\item $\sfn_{\alpha} \geq \sfr_{\alpha}$ for all $\alpha \in [0, 1]$. 
\item $\int_{\alpha = 0}^1 \sfn_{\alpha} = 1$. 
\item There exists $t > 0$ such that for all $\alpha \in [0, 1]$, one of the following must hold, because otherwise we can decrease one $n_{\alpha}$ and increase another $n_{\alpha'}$ to further decrease \eqref{eq:minsumff}. Note that $\frac{df_1(\sfn, \sfr)}{d\sfn} = 2\sfn - \sfr/2$ and 
$\frac{df_2(\sfn, \sfr)}{d\sfn} = \sfn + \sfr$. 

\begin{itemize}
\item If $\sfn_{\alpha} = \sfr_{\alpha}$, $\frac{df(\na, \ra)}{d\na} = \frac{df_2(\na, \ra)}{d\na} = 2\na \geq t$. 
\item If $\sfn_{\alpha} = 2\sfr_{\alpha}$, 
$\frac{df(\na^{-}, \ra)}{d\na} = \frac{df_2(\na, \ra)}{d\na} = 1.5\na \leq t$ and 
$\frac{df(\na^{+}, \ra)}{d\na} = \frac{df_1(\na, \ra)}{d\na} = 1.75\na \geq t$. 
\item Otherwise, $\frac{f(\sfn_{\alpha}, \sfr_{\alpha})}{d\sfn_{\alpha}} = t$. 
\end{itemize}

\end{enumerate}
It is easy to see that $t < 2$, because otherwise $\sfn_{\alpha} > 1$ for all $\alpha \in (0, 1]$, violating the condition 2. 
This implies that $t = 2\exp(-c)$ for some $c > 0$ to be determined and
\begin{equation}
\sfn_{\alpha} = \sfr_{\alpha} \mbox{ for all } \alpha \in [0, c].
\label{eq:minsum1}
\end{equation}
Since $f(\sfn, \sfr) = f_2(\sfn, \sfr)$ when $\sfn \leq 2\sfr$ and $f_1(\sfn, \sfr)$ otherwise, 
to meet the condition 3, we have the following conditions.

\noindent Whenever $\ra < \na < 2\ra$ (which implies $f(\na, \ra) = f_2 (\na, \ra)$), 
\begin{equation}
\frac{df_2(\na, \ra)}{\na} = \na + \ra = t \Rightarrow \na = 2\exp(-c) - \exp(-\alpha).
\label{eq:minsum2}
\end{equation}
Whenever $\na > 2\ra$ (which implies $f(\na, \ra) = f_2 (\na, \ra)$), 
\begin{equation}
\frac{df_1(\na, \ra)}{\na} = 2\na - \ra/2 = t \Rightarrow \na = \exp(-c) + \exp(-\alpha) / 4. 
\label{eq:minsum3}
\end{equation}
To meet~\eqref{eq:minsum1},~\eqref{eq:minsum2}, and~\eqref{eq:minsum3}, 
$n_{\alpha}$ has to be 
\begin{equation}
n_{\alpha} =
\begin{cases}
\exp(-\alpha), & \alpha \in [0, c] \\
2\exp(-c) - \exp(-\alpha), & \alpha \in [c, d_1] \\
2\exp(-\alpha), & \alpha \in [d_1, d_2] \\
\exp(-c) + \exp(-\alpha) / 4, & \alpha \in [d_2, 1]
\end{cases}
\end{equation}
where $d_1 = \ln(3/2) + c$ and $d_2 = \ln(7/4) + c$ so that $f(\sfn_{\alpha}, \sfr_{\alpha}) = f_2(\sfn_{\alpha}, \sfr_{\alpha})$ for $\alpha \in [c, d_1]$ and 
$f(\sfn_{\alpha}, \sfr_{\alpha}) = f_1(\sfn_{\alpha}, \sfr_{\alpha})$ for $\alpha \in [d_2, 1]$. ($f_1(\sfn_{\alpha}, \sfr_{\alpha}) = f_2(\sfn_{\alpha}, \sfr_{\alpha})$ when $\alpha \in [d_1, d_2]$.)
Then 
\begin{align*}
& \int_{\alpha=0}^1 n_{\alpha} d\alpha \\
=& \int_{\alpha=0}^c e^{-\alpha} d\alpha 
	+ \int_{\alpha=c}^{d_1} (2e^{-c} - e^{-\alpha}) d\alpha 
	+ \int_{\alpha=d_1}^{d_2} 2 e^{-\alpha} d\alpha 
	+ \int_{\alpha=d_2}^1 (e^{-c} + e^{-\alpha} / 4) d\alpha \\
=& \bigg( 1-e^{-c}  \bigg) + 
	\bigg( 2\ln(3/2)e^{-c} - e^{-c} + e^{-d_1}  \bigg) + 
	2 \bigg( e^{-d_1} - e^{-d_2}  \bigg) + 
	\bigg( (1 - d_2)e^{-c} + (e^{-d_2} - e^{-1}) / 4  \bigg) \\
=& 1 - e^{-1}/4 + e^{-c} ( \ln (9/7) - c ) = 1.
\end{align*}
Where the third equality uses the definitions of $d_1$ and $d_2$. 
This implies $e^{-c} ( \ln (9/7) - c ) = e^{-1} / 4$, which solves to $c = \ln(9/7) - W(9/28e) \approx 0.145$ where $W(z)$ denotes the real solution of $z = We^W$. Plugging this value into 
\[
\int_{\alpha=0}^1 f(\na, \ra) d\alpha 
= \int_{\alpha=0}^{d_2} f_2(\na, \ra) \alpha
+ \int_{\alpha=d_2}^{1} f_1(\na, \ra) \alpha
\]
gives $\geq 0.7079$. Therefore, the \minsum cost in the soundness case is at least $(0.7079 - o(1)) nr$. 
Compared to the cost $(1/2 + o(1))nr$ in the completeness case, the gap is $\geq \minsumvalue$.
\end{proof}

\iffalse
Given $r_i = r \cdot e^{-(k-i)/k}$, the worst choice $n_i$ should satisfy $\sum_{i=1}^k n_i = (1 - o(1))nk$, and 
\[
\frac{n_i ( n_i - \frac{r_i}{2} )}{d n_i} = 2n_i - \frac{r_i}{2}
\]
is constant for all $i$. This is satisfied by $n_i = \frac{3 + 1/e}{4}\cdot r + r_i / 4$, that yields
\begin{align*}
&\sum_{i = 1}^{k} n_i \bigg( n_i - \frac{r_i}{2} \bigg) 
= 
\sum_{i = 1}^{k} \bigg( \frac{(3 + 1/e)r}{4}+\frac{r_i}{4} \bigg) \bigg(\frac{(3 + 1/e)r}{4} - \frac{r_i}{4}\bigg)  \\
=& 
(1 - o(1)) nr \int_{\alpha = 0}^1 
\bigg( \frac{3 + 1/e}{4}+\frac{e^{-\alpha}}{4} \bigg) \bigg(\frac{3 + 1/e}{4} - \frac{e^{-\alpha}}{4}\bigg) \\
=& 
(1 - o(1)) nr \int_{\alpha = 0}^1 
\bigg( \frac{9 + 6/e + 1/e^2}{16} - \frac{e^{-2\alpha}}{16} \bigg)  \\
=& 
(1 - o(1)) nr \cdot \frac{17 + 12/e + 3/e^2}{32}.
\end{align*}

Therefore, the gap between the completeness case and the soundness case is 
\[
\frac{17 + 12/e + 3/e^2}{16} \approx 1.3637.
\]
\fi

\section{Inapproximability of Continuous \kmean and \kmed in $\ell_\infty$-metric}\label{sec:cont}

In this section, we prove the highest inapproximability factor known for \kmean and \kmed  in literature (in any metric), i.e., we prove Theorem~\ref{thm:introellinf}. The proof relies crucially on the following result of Khot and Saket.

\begin{theorem}[Khot and Saket \cite{KS12}]\label{thm:KS12}
For any constant $\varepsilon>0$, and positive integers $t$ and $q$ such that $q \ge 2^t + 1$, given a graph $G(V, E)$, it is \NP-hard to distinguish between the following two cases:
\begin{itemize}
\item \textbf{\emph{Completeness}}:  There are $q$ disjoint independent sets $V_1,\ldots ,V_q \subseteq V$, such that for all $i\in[q]$ we have $|V_i| = \frac{(1-\varepsilon)}{q}\cdot |V|$.
\item \textbf{\emph{Soundness}}: There is no independent set in $G$ of size $\frac{1}{q^{t+1}}\cdot |V|$.
\end{itemize}
\end{theorem}

We are now ready to prove the main result of this section. 

\begin{theorem}[\kmean without candidate centers in $n^{O(1)}$ dimensional $\ell_\infty$-metric space]\label{thm:kmeanellinfhighdim} For any constant $\varepsilon>0$ and any constant $\alpha\in\NN$, there exists a constant $k:=k(\eps,\alpha)\in\mathbb{N}$, such that 
given a point-set $\Po\subset \R^m$ of size $n$ (and $m=\poly(n)$), it is \NP-hard to distinguish between the following two cases:
\begin{itemize}
\item \textbf{\emph{Completeness}}:  
There exists $\C':=\{c_1,\ldots ,c_k\}\subseteq \mathbb R^m$ and $\sigma:\Po\to\C'$ such that $$\sum_{a\in\Po}\left(\|a-\sigma(a)\|_\infty\right)^2\le (1+8\varepsilon)\cdot n,$$
\item \textbf{\emph{Soundness}}: For every $\C':=\{c_1,\ldots ,c_{\alpha k}\}\subseteq \mathbb R^m$ and every $\sigma:\Po\to\C'$ we have: $$\sum_{a\in\Po}\left(\|a-\sigma(a)\|_\infty\right)^2\ge (4-\varepsilon)\cdot n.$$
\end{itemize}
\end{theorem}

\begin{theorem}[\kmed without candidate centers in $n^{O(1)}$ dimensional $\ell_\infty$-metric space]\label{thm:kmedellinfhighdim} For any constant $\varepsilon>0$ and any constant $\alpha\in\NN$, there exists a constant $k:=k(\eps,\alpha)\in\mathbb{N}$, such that  
given a point-set $\Po\subset \R^m$ of size $n$ (and $m=\poly(n)$), it is \NP-hard to distinguish between the following two cases:
\begin{itemize}
\item \textbf{\emph{Completeness}}:  
There exists $\C':=\{c_1,\ldots ,c_k\}\subseteq \mathbb R^m$ and $\sigma:\Po\to\C'$ such that $$\sum_{a\in\Po}\|a-\sigma(a)\|_\infty\le (1+2\varepsilon)\cdot n,$$
\item \textbf{\emph{Soundness}}: For every $\C':=\{c_1,\ldots ,c_{\alpha k}\}\subseteq \mathbb R^m$ and every $\sigma:\Po\to\C'$ we have: $$\sum_{a\in\Po}\|a-\sigma(a)\|_\infty\ge (2-\varepsilon)\cdot n.$$
\end{itemize}
\end{theorem}

\begin{proof}[Proof of Theorems~\ref{thm:kmeanellinfhighdim} and \ref{thm:kmedellinfhighdim}]
Fix $\varepsilon>0$ as in the theorem statement. Let $r=\alpha k$ and $\eps':=\eps/r$. Starting from the hard instance $(G(V,E),q,t,\eps')$ given in Theorem~\ref{thm:KS12},
  we create an instance of the \kmean, or of the \kmed problem, where $k=q$ (and $t=o(\log k)$),  as follows.

\paragraph{Construction.}  The \kmed or \kmean instance consists of the set of  points to be clustered $\Po\subseteq \R^{m}$ of size $n$ (where $n=|V|,\ m=|E|$) which will be defined below. 
 First, we arbitrarily orient the edges of $G$ (so that for every $(u,v)\in V\times V$, at most one of $(u,v)$ or $(v,u)$ is in $E$). Then, we will construct function $A:V\to\mathbb{R}^{m}$.  Given $A$, the point-set $\Po$ is just defined to be $$\Po:=\left\{ A(v)\big| v\in V\right\}.$$

For every $v\in V$ and every $(u',v')\in E$, we define the $(u',v')^{\text{th}}$ coordinate of $A(v)$ as follows
$$
A(v)_{(u',v')}:=\begin{cases}
2\text{\ \ \ \ \ \  if }v=u'\\
-2\text{\ \ \ if }v=v'\\
0\text{\ \ \ \ \ \  otherwise}
\end{cases}.
$$

  We now analyze the \kmean and \kmed cost of the
  instance.
  Consider the completeness case first.

  \paragraph*{Completeness.}
Suppose there are $k$ disjoint independent sets $V_1,\ldots ,V_k \subseteq V$, such that for all $i\in[k]$ we have $|V_i| = \frac{(1-\varepsilon')}{k}\cdot |V|$. Then, we partition $\Po$ into $k$ clusters,  say $C_1,\ldots ,C_k$, as follows. For every $p\in\Po $ where $p:=A(v)$ for some $v\in V$, if there is some $i\in[k]$ such that $v\in V_i$ then we assign $p$ to cluster $C_i$; otherwise, we assign it to cluster $C_1$. Next, we define the cluster centers $\C=\{c_1,\ldots ,c_k\}\subseteq \R^m$ as follows. For every $i\in[k]$, and every $(u',v')\in E$, the $(u',v')^{\text{th}}$ coordinate of $c_i$ is defined as follows
$$
c_i(u',v'):=\begin{cases}
1\text{\ \ \ \ \ \  if }u'\in V_i\\
-1\text{\ \ \ if }v'\in V_i\\
0\text{\ \ \ \ \ \  otherwise}
\end{cases}.
$$

Note that the definition of the $(u',v')^{\text{th}}$ coordinate of $c_i$ is consistent, as $V_i$ is an independent set and thus both $u'$ and $v'$ cannot be in $V_i$. For any $p\in \Po$ and any $c\in\C$, we have the following upper bound on their distance:
\begin{align}\label{equb}
\|p-c\|_\infty\le 3.
\end{align}

On the other hand for every $i\in[k]$, and every $v\in V_i$, we have the following computation on distance of $A(v)$ to its center.  
\begin{align}\label{eqb}
\|A(v)-c_i\|_\infty= \max\left\{\underset{(u,v)\in E}{\max}|A(v)_{(u,v)}+1|,\underset{(v,u)\in E}{\max}|A(v)_{(v,u)}-1| ,\underset{\substack{e\in E\\ v\notin e}}{\max}|c_i(e)|\right\}=1.
\end{align}

Therefore,  from \eqref{eqb}, the \kmean and \kmed cost of cluster $C_i$ for all $i\in [k]\setminus \{1\}$ is exactly $|V_i|$. On the other hand, putting together \eqref{equb} and \eqref{eqb}, the \kmean cost of $C_1$ is upper bounded by:
$$
|V_1|+9\cdot \left(|V|-\sum_{i\in[k]}|V_i|\right)\le |V_1|+9\varepsilon'|V|.
$$
Similarly, we have that the \kmed cost of $C_1$ is upper bounded by $
|V_1|+3\varepsilon'|V|$.

Thus,  the \kmean cost of the overall instance is at most $ |V|(1+8\varepsilon')$, while the \kmed cost is
  $ |V|(1+2\varepsilon')$. 
   Finally, we turn to the soundness analysis. 
  
    \paragraph*{Soundness.} We have that from the soundness case assumption that every subset $S\subset V$ of size at least $\eps'|V|$ is not an independent set in $G$.  Consider any set of centers $\C' = \{c_1,\ldots,c_r\}\subset \R^m$
  that is optimal for the \kmed or \kmean objective (and let $C_1,\ldots ,C_r$ be the corresponding partitioning of $\Po$ into $r$ clusters). 
  We have the following claim.
  \begin{claim}\label{cl:so}
  Let $i\in[r]$ and $V_i:=\{v\in V\mid A(v)\in C_i\}$. Then, there are $\nicefrac{\left(|V_i|-\eps'|V|\right)}{2}$ vertex disjoint edges in the induced subgraph of $V_i$ in $G$.
  \end{claim}
  \begin{proof}
  Suppose $|V_i|\ge \eps'|V|$ then there exists an edge in the induced subgraph of $V_i$ in $G$. Remove the two corresponding vertices of the edge  from $V_i$. Repeat the above procedure until $|V_i|<\eps'|V|$. The vertex pairs (which are edges in $G$) that were removed  would be at least $\nicefrac{\left(|V_i|-\eps'|V|\right)}{2}$ in number.
  \end{proof}
  
For every $i\in[r]$, let $E_i$ be the set of vertex disjoint edges guaranteed by the above claim.   Fix $i\in[r]$. For every $e:=(u',v')\in E_i$ we have:
\begin{align}\label{eqsoundmed}
\|A(u')-c_i\|_\infty+\|A(v')-c_i\|_\infty\ge \|A(v')-A(u')\|_\infty\ge|A(u')_e-A(v')_e|\ge 4.
\end{align}
  We also have:
\begin{align}
\|A(u')-c_i\|_\infty^2+\|A(v')-c_i\|_\infty^2&\ge (A(u')_e-c_i(e))^2+(A(v')_e-c_i(e))^2\nonumber\\
&\ge\frac{1}{2}\cdot (A(u')_e-A(v')_e)^2\ge 8\label{eqsoundmean}.
\end{align}  

  Therefore, the optimal solution w.r.t. \kmed objective has cost at least:
  \begin{align*}\allowdisplaybreaks
  &\sum_{i\in[r]}\sum_{v\in V_i}\|A(v)-c_i\|_\infty&\\
  &\ge \sum_{i\in[r]}\sum_{(u',v')\in E_i}(\|A(u')-c_i\|_\infty+\|A(v')-c_i\|_\infty)&\\
  &\ge \sum_{i\in[r]}(4\cdot |E_i|)&(\text{ from }\eqref{eqsoundmed})\\
&\ge   \sum_{i\in[r]}\left(2\left(|V_i|-\eps'|V|\right)\right)&(\text{ from Claim~\ref{cl:so}})\\
&\ge (2-\eps'r)\cdot |V|=(2-\eps)\cdot |V|&
  \end{align*}
  
  Similarly, the optimal solution w.r.t. \kmean objective has cost at least:\allowdisplaybreaks
  \begin{align*}\allowdisplaybreaks
  &\sum_{i\in[r]}\sum_{v\in V_i}\|A(v)-c_i\|_\infty^2&\\
  &\ge \sum_{i\in[r]}\sum_{(u',v')\in E_i}(\|A(u')-c_i\|_\infty^2+\|A(v')-c_i\|_\infty^2)&\\
  &\ge \sum_{i\in[r]}(8\cdot |E_i|)&(\text{ from }\eqref{eqsoundmean})\\
&\ge   \sum_{i\in[r]}\left(4\left(|V_i|-\eps'|V|\right)\right)&(\text{ from Claim~\ref{cl:so}})\\
&\ge (4-\eps'r)\cdot |V|=(4-\eps)\cdot |V|&
  \end{align*}
 \end{proof}

To prove that  Theorems~\ref{thm:kmeanellinfhighdim}~and~\ref{thm:kmedellinfhighdim}  hold even when in the completeness case we have $k=4$, we simply start from the below theorem instead of Theorem~\ref{thm:KS12}. 

\begin{theorem}[\cite{KMS17,DKKMS18a,DKKMS18b,BKS19,KMS18}]
For any constant $\varepsilon>0$, given a graph $G(V, E)$, it is \NP-hard to distinguish between the following two cases:
\begin{itemize}
\item \textbf{\emph{Completeness}}:  There are $4$ disjoint independent sets $V_1,V_2,V_3,V_4 \subseteq V$, such that $|V_1| = |V_2|=|V_3| = |V_4|= \frac{(1-\varepsilon)}{4}\cdot |V|$.
\item \textbf{\emph{Soundness}}: There is no independent set in $G$ of size $\eps\cdot |V|$.
\end{itemize}
\end{theorem}

We remark that Theorem~\ref{thm:introellinf} can also be obtained for $\ell_p$-metrics as $p$ tends to $\infty$. An interesting variant of Theorems~\ref{thm:kmeanellinfhighdim}~and~\ref{thm:kmedellinfhighdim}, is when we restrict that the centers have to be picked from $\Z^d$ (where $d=\poly(n)$) instead of allowing to pick them from anywhere in $\R^d$. This can be seen as in between the traditional discrete and continuous case, where the size of the set of candidate centers is exponential in the number of points to be clustered, but has a compact representation (in this case fixed representation depending only on $n$). Surprisingly, for this variant, we show even stronger inapproximability factors of $9-\eps$ for \kmean and $3-\eps$ for \kmed (see Theorems~\ref{thm:kmeanellinflattice}~and~\ref{thm:kmedellinflattice} in Appendix~\ref{sec:lattice}), for any small $\eps>0$. We prove below a strengthening of Theorems~\ref{thm:kmeanellinfhighdim}~and~\ref{thm:kmedellinfhighdim} under the unique games conjecture.

\begin{theorem}[Bi-criteria 2-$\mathsf{mean}$ and 2-$\mathsf{median}$ without candidate centers in $n^{O(1)}$ dimensional $\ell_\infty$-metric space]\label{thm:2meanmedellinfhighdim} Assuming the unique games conjecture, for any constant $\varepsilon>0$, and every constant $r\in\mathbb{N}$,
given a point-set $\Po\subset \R^m$ of size $n$ (and $m=\poly(n)$), it is \NP-hard to distinguish between the following two cases:
\begin{itemize}
\item \textbf{\emph{Completeness}}:  
There exists $\C':=\{c_1,c_2\}\subseteq \mathbb R^m$ and $\sigma:\Po\to\C'$ such that $$\sum_{a\in\Po}\left(\|a-\sigma(a)\|_\infty\right)^2\le n\ \ \  \left(\text{resp.\ }\sum_{a\in\Po}\|a-\sigma(a)\|_\infty\le n\right),$$
\item \textbf{\emph{Soundness}}: For every $\C':=\{c_1,\ldots ,c_r\}\subseteq \mathbb R^m$ and every $\sigma:\Po\to\C'$ we have: $$\sum_{a\in\Po}\left(\|a-\sigma(a)\|_\infty\right)^2\ge (4-\varepsilon)\cdot n\ \ \ \left(\text{resp.\ }\sum_{a\in\Po}\|a-\sigma(a)\|_\infty\ge (2-\varepsilon)\cdot n\right).$$
\end{itemize}
\end{theorem}

The  proof simply follows by using the  following result of Bansal and Khot instead of Theorem~\ref{thm:KS12}.

\begin{theorem}[Bansal and Khot \cite{BK09}]
Assuming the unique games conjecture, for any constant $\varepsilon>0$, given a graph $G(V, E)$, it is \NP-hard to distinguish between the following two cases:
\begin{itemize}
\item \textbf{\emph{Completeness}}:  There are $2$ disjoint independent sets $V_1,V_2 \subseteq V$, such that $|V_1| = |V_2|= \frac{(1-\varepsilon)}{2}\cdot |V|$.
\item \textbf{\emph{Soundness}}: There is no independent set in $G$ of size $\eps\cdot |V|$.
\end{itemize}
\end{theorem}

\subsection{Approximability}
We now show that the above bound is tight for a large range of settings.
First, for any $k$, there is an algorithm running in time $dn^{k+2}$
that takes as input a set of points in $\R^d$ and output a 2-approximate solution to the
continuous \kmed problem (and a 4-approximation solution for the continuous \kmean problem)
in the $\ell_{\infty}$-metric (see Fact~\ref{fact:2approx}). 
Second, we show how to obtain a $(1+\eps)$-approximation solution in time
$(kd\eps^{-1} \log n)^{O(k)}(1/\eps)^{O(dk)} + \text{poly}(nd/\eps)$ (see Corollary~\ref{cor:below2}).
Third, we show a $(2+\eps)$-approximation solution in time $O((\eps^{-1}kd\log n)^{O(k)} + (nd)^{O(1)})$
which is fixed parameter tractable when parameterized by $k$, for any $d = 2^{O(\log^{1-\delta} (n))}$, where $\delta$ is a constant less than 1
(see Corollary~\ref{cor:fptk}).
%% \vnote{Maybe we can do:
Finally, we provide an $(1+2/e+\eps)$-approximate solution in time
$(kd\eps^{-1} \log n)^{O(k)} +(kd\eps^{-1} \log n)^{O(1)}(1/\eps)^{O(d)} + \text{poly}(nd/\eps)$
which shows that for the hardness bounds mentioned above, the dependency in $d$
cannot be significantly improved unless $k$ becomes large (see Corollary~\ref{cor:1+2/e}).

\begin{fact}
  \label{fact:2approx}
  There exists a 2-approximation algorithm (resp. 4-approximation algorithm) that
  for any   instance of the continuous \kmed (resp. \kmean) problem  consisting of
  $n$ points $\Po$ in $\R^d$ in the $\ell_{\infty}$-metric runs in time $dn^{k+2}$.
\end{fact}
\begin{proof}
  Consider an instance of the continuous \kmed problem consisting of a set of $n$ points
  in $\R^d$ (an analogous argument applies to the \kmean problem).
  Consider the solution $\tilde{S}$ obtained from the optimal solution as follows: for each center
  $c_i$ of the optimal solution, pick the point $p_{c_i}$ of $\Po$ that is the closest to $c_i$. $\tilde{S}$
  obviously contains at most $k$ centers and so is a valid solution.
  Now, each point $p \in \Po$ whose closest center in the optimal solution
  is $c_i$ has a center that is no further away than $p_{c_i}$. Since by the choice of $p_{c_i}$ we have
  that $||p-c_i||_\infty \ge ||p_{c_i}-c_i||_{\infty}$, and we have by the triangle inequality
  $||p-p_{c_i}||_{\infty} \le 2||p-c_i||_\infty$ and so $\tilde{S}$ is at most a 2-approximation.

  Thus, the algorithm that enumerates all possible $k$-tuples of $\Po$ and outputs the one that induces
  the minimum \kmed cost achieves a $2$-approximation in the above time bound.
\end{proof}

We then turn to the following fact which states that up to losing a $(1+\eps)$-factor in the approximation
guarantee, one can identify a discrete set of centers of size at most $n(1/\eps)^{O(d)}\log n$.
Given an instance $\Po$ of the continuous \kmed (resp. \kmean problem),
we define an $\eps$-approximate candidate center set for $\Po$ as a set $\C$ such that 
there exists a set of $k$ points of $\C$ whose \kmed (resp. \kmean) cost
is at most $(1+\eps)$ times the cost of the optimal continuous \kmed (resp. \kmean) clustering.
\begin{lemma}
  \label{lem:simple}
  There exists an algorithm that takes as input an instance $\Po$ of the continuous \kmed
  (resp. continuous \kmean) in $\R^d$ and that produces an $\eps$-approximate candidate center set $\C$ of size
  $|\Po|(1/\eps)^{O(d)}\log |\Po|$.
\end{lemma}
\begin{proof}
  The proof follows from designing approximate candidate center sets (see~\cite{Matousek00,Cohen-AddadL19}
  for similar results for the
  $\ell_2$-metric). Let $n = |\Po|$.
  The set of candidate centers $\C$ is iteratively constructed as follows. Let $\gamma$ be an estimate of
  the cost of the optimal solution (which can be computed in polynomial time using an $O(1)$-approximate
  solution on the discrete version of the problem where the set of candidate centers is $\Po$;
  Fact~\ref{fact:2approx} guarantees that it is an $O(1)$-approximate solution to the continuous version).
  First start with $\C = \Po$. Then, for each point $p \in S$,
  for each $2^i$ such that $\eps\gamma/n \le 2^i \le 2\gamma$, consider the ball of center $p$ and
  radius $2^i$ and pick an $\eps\cdot 2^i$-net in this ball, the size of the net is at most
  $(1/\eps)^{O(d)}$. Add the net to $\C$.
  
  The total size of the candidate center set $\C$ follows immediately from the definition.
  We thus turn to proving the correctness. Consider the optimal solution and let's build
  a solution $S \subseteq \C$ of cost at most $(1+\eps)$ times higher. For any center $c$ in the optimal
  solution, consider the closest point $p_c$ in $\Po$ and let $\delta$ be $||p_c - c||_{\infty}$.
  Let $\tilde{c}$ be the point of $\C$ that is the closest to $c$. By triangle inequality and the definition
  of the net, we have that $||\tilde{c}-c||_{\infty} < \eps \delta$.
  Therefore, applying the triangle inequality, each point in the cluster of $c$ can be assigned
  to $\tilde{c}$ at an additive cost increase of $\eps\delta$. Moreover, since each point of the cluster
  is at distance at least $\delta$ from $c$, the cost to assign each point in cluster $c$ to $\tilde{c}$ is
  no more than $(1+\eps)$ times higher than the cost of assigning these points to $c$ and so follows the lemma.
\end{proof}

For proving Corollaries~\ref{cor:below2},~\ref{cor:fptk},~and~\ref{cor:1+2/e}, we will make use of the notion of coreset.
A (strong) $\eps$-\emph{coreset} for a discrete \kmed instance of $n$ points $\Po$ and $m$ candidate centers $\C$ is a
set of points $W$ with a weight function $w: W \mapsto \R^+$ such that for any set of centers $S \subseteq \C$ of size $k$, we have:
\[\sum_{p \in P_0}\min_{s \in S} \text{dist}(p,s) = (1\pm\eps) \sum_{p \in W}w(p) \min_{s \in S} \text{dist}(p,s).\]

\begin{sloppypar}We now consider the following lemma from Langberg and Feldman~\cite{feldman2011unified} and Chen~\cite{chen2009coresets}.\end{sloppypar}

\begin{lemma}[\cite{feldman2011unified,chen2009coresets} -- Restated]
  \label{lem:coreset}
  There exists a polynomial-time algorithm that 
  on any instance of the discrete \kmed problem consisting of $n$ points and
  $m$ candidate centers, outputs an $\eps$-coreset of size $(k\eps^{-1}\log m)^{O(1)}$.
\end{lemma}

From there we can deduce the following corollary.
\begin{corollary}
  \label{cor:below2}
  There exists a 2-approximation algorithm for continuous \kmed instances of $n$ points in $\R^d$
  with running time 
  $(kd\eps^{-1} \log n)^{O(k)} + \text{poly}(nd/\eps)$.
\end{corollary}
\begin{proof}
  The corollary follow from Lemma~\ref{lem:simple} and Lemma~\ref{lem:coreset}: one can obtain
  an $\eps$-coreset $C_0$ of size $(kd\eps^{-1} \log n)^{O(1)}$ of any \kmed instance consisting
  of $n$ points in $\R^d$. Hence, by Fact~\ref{fact:2approx}, the best \kmed solution
  whose centers are in $C_0$ is a $(2+\eps)$-approximation to the original continuous \kmed instance and
  so, the algorithm that enumerates all $k$-tuples of $C_0$ and outputs the one that has minimum \kmed
  cost for the instance achieves a $(2+\eps)$-approximation in the prescribed time bounds.
\end{proof}

\begin{corollary}
  \label{lem:candidatecenters}
  There exists an algorithm that on any continuous \kmed instance of $n$ points in $\R^d$,
  produces an $\eps$-approximate candidate center set of size $(kd\eps^{-1} \log n)^{O(1)}(1/\eps)^{O(d)}$.
\end{corollary}
\begin{proof}
  The proof follows from applying Lemma~\ref{lem:coreset} on the input
  points and the $\eps$-approximate candidate center set $\C$ described by
  Lemma~\ref{lem:simple}.
  Then, by observing that the proof of Lemma~\ref{lem:simple} also applies to weighted
  set of points, one can further reduce the number of candidate centers to a set
  $\C'$ of size
  $(kd\eps^{-1} \log n)^{O(1)} (1/\eps)^{O(d)}$.
\end{proof}

\begin{corollary}
  \label{cor:fptk}
  There exists a $(1+\eps)$-approximation algorithm with running time
  \[(kd\eps^{-1} \log n)^{O(k)}(1/\eps)^{O(dk)} + \text{poly}(nd/\eps).\]
\end{corollary}
\begin{proof}
  The $(1+\eps)$-approximation algorithm follows from computing the set of candidate centers
  $\C'$ prescribed by Corollary~\ref{lem:candidatecenters} and
  enumerating all $k$-tuples of
  $\C'$ and outputting the one which induces the smallest \kmed cost.
\end{proof}

\begin{corollary}
  \label{cor:1+2/e}
  There exists a $(1+2/e+\eps)$-approximation algorithm for continuous \kmed instances of $n$ points in $\R^d$
  with running time
  \[(kd\eps^{-1} \log n)^{O(k)} +(kd\eps^{-1} \log n)^{O(1)}(1/\eps)^{O(d)} + \text{poly}(nd/\eps).\]
\end{corollary}
\begin{proof}
  Applying Corollary~\ref{lem:candidatecenters}, one constructs an instance of the discrete \kmed
  instance with $n = |\Po|$ points and $m = (kd\eps^{-1} \log n)^{O(1)}(1/\eps)^{O(d)}$ candidate centers.
  Then, one can compute a $(1+2/e+\eps)$-approximation to this instance in time
  $(k \eps^{-1}\log m \log n)^{O(k)} + m$ using the FPT algorithm of~\cite{Cohen-AddadG0LL19}.
\end{proof}

\subsection*{Acknowledgements}
We are truly grateful to Pasin Manurangsi for various detailed discussions that inspired many of the results in this paper.

%% Ce projet a b\'en\'efici\'e d'une aide de l'\'Etat g\'er\'ee
%% par l'Agence Nationale de la Recherche au titre du Programme
%% Appel à projets générique JCJC 2018 portant la r\'ef\'erence
%% suivante : ANR-18-CE40-0004-01.
Karthik C.\ S.\ was supported by Irit Dinur's ERC-CoG grant 772839, the  Israel Science Foundation (grant number 552/16), the Len Blavatnik and the Blavatnik Family foundation, and Subhash Khot's Simons Investigator Award. 
Euiwoong Lee was supported in part by the Simons Collaboration on Algorithms and Geometry.

\bibliographystyle{alpha}
\bibliography{references}

\newcommand{\etalchar}[1]{$^{#1}$}
\begin{thebibliography}{dlVKKR03}

\bibitem[ACKS15]{ACKS15}
Pranjal Awasthi, Moses Charikar, Ravishankar Krishnaswamy, and Ali~Kemal Sinop.
\newblock The hardness of approximation of euclidean k-means.
\newblock In {\em 31st International Symposium on Computational Geometry, SoCG
  2015, June 22-25, 2015, Eindhoven, The Netherlands}, pages 754--767, 2015.

\bibitem[ANSW20]{ANSW16}
Sara Ahmadian, Ashkan Norouzi{-}Fard, Ola Svensson, and Justin Ward.
\newblock Better guarantees for k-means and euclidean k-median by primal-dual
  algorithms.
\newblock {\em {SIAM} J. Comput.}, 49(4), 2020.

\bibitem[BCR01]{bartal2001approximating}
Yair Bartal, Moses Charikar, and Danny Raz.
\newblock Approximating min-sum \emph{k}-clustering in metric spaces.
\newblock In {\em Proceedings on 33rd Annual {ACM} Symposium on Theory of
  Computing, July 6-8, 2001, Heraklion, Crete, Greece}, pages 11--20, 2001.

\bibitem[BFSS19]{Behsaz2019}
Babak Behsaz, Zachary Friggstad, Mohammad~R. Salavatipour, and Rohit Sivakumar.
\newblock Approximation algorithms for min-sum k-clustering and balanced
  k-median.
\newblock {\em Algorithmica}, 81(3):1006--1030, Mar 2019.

\bibitem[BG95]{bronnimann1995almost}
Herv{\'e} Br{\"o}nnimann and Michael~T Goodrich.
\newblock Almost optimal set covers in finite vc-dimension.
\newblock {\em Discrete \& Computational Geometry}, 14(4):463--479, 1995.

\bibitem[BK09]{BK09}
Nikhil Bansal and Subhash Khot.
\newblock Optimal long code test with one free bit.
\newblock In {\em 50th Annual {IEEE} Symposium on Foundations of Computer
  Science, {FOCS} 2009, October 25-27, 2009, Atlanta, Georgia, {USA}}, pages
  453--462, 2009.

\bibitem[BKL12]{badanidiyuru2012approximating}
Ashwinkumar Badanidiyuru, Robert Kleinberg, and Hooyeon Lee.
\newblock Approximating low-dimensional coverage problems.
\newblock In {\em Proceedings of the twenty-eighth annual symposium on
  Computational geometry}, pages 161--170. ACM, 2012.

\bibitem[BKS19]{BKS19}
Boaz Barak, Pravesh~K. Kothari, and David Steurer.
\newblock Small-set expansion in shortcode graph and the 2-to-2 conjecture.
\newblock In {\em 10th Innovations in Theoretical Computer Science Conference,
  {ITCS} 2019, January 10-12, 2019, San Diego, California, {USA}}, pages
  9:1--9:12, 2019.

\bibitem[BPR{\etalchar{+}}15]{BPRST15}
Jaroslaw Byrka, Thomas Pensyl, Bartosz Rybicki, Aravind Srinivasan, and Khoa
  Trinh.
\newblock An improved approximation for \emph{k}-median, and positive
  correlation in budgeted optimization.
\newblock In {\em Proceedings of the Twenty-Sixth Annual {ACM-SIAM} Symposium
  on Discrete Algorithms, {SODA} 2015, San Diego, CA, USA, January 4-6, 2015},
  pages 737--756, 2015.

\bibitem[CGK{\etalchar{+}}19]{Cohen-AddadG0LL19}
Vincent Cohen{-}Addad, Anupam Gupta, Amit Kumar, Euiwoong Lee, and Jason Li.
\newblock Tight {FPT} approximations for k-median and k-means.
\newblock In {\em 46th International Colloquium on Automata, Languages, and
  Programming, {ICALP} 2019, July 9-12, 2019, Patras, Greece.}, pages
  42:1--42:14, 2019.

\bibitem[Che09]{chen2009coresets}
Ke~Chen.
\newblock On coresets for k-median and k-means clustering in metric and
  euclidean spaces and their applications.
\newblock {\em SIAM Journal on Computing}, 39(3):923--947, 2009.

\bibitem[CK19]{CK19}
Vincent {Cohen-Addad} and {Karthik {C. S.}}
\newblock Inapproximability of clustering in $l_p$-metrics.
\newblock In {\em 2019 IEEE 60th Annual Symposium on Foundations of Computer
  Science}, pages 519--539, 2019.

\bibitem[CL19]{Cohen-AddadL19}
Vincent Cohen{-}Addad and Jason Li.
\newblock On the fixed-parameter tractability of capacitated clustering.
\newblock In {\em 46th International Colloquium on Automata, Languages, and
  Programming, {ICALP} 2019, July 9-12, 2019, Patras, Greece.}, pages
  41:1--41:14, 2019.

\bibitem[CS04]{czumaj2004sublinear}
Artur Czumaj and Christian Sohler.
\newblock Sublinear-time approximation for clustering via random sampling.
\newblock In {\em International Colloquium on Automata, Languages, and
  Programming}, pages 396--407. Springer, 2004.

\bibitem[CS10]{czumaj2010small}
Artur Czumaj and Christian Sohler.
\newblock Small space representations for metric min-sum k-clustering and their
  applications.
\newblock {\em Theory of Computing Systems}, 46(3):416--442, 2010.

\bibitem[DKK{\etalchar{+}}18a]{DKKMS18b}
Irit Dinur, Subhash Khot, Guy Kindler, Dor Minzer, and Muli Safra.
\newblock On non-optimally expanding sets in grassmann graphs.
\newblock In {\em Proceedings of the 50th Annual {ACM} {SIGACT} Symposium on
  Theory of Computing, {STOC} 2018, Los Angeles, CA, USA, June 25-29, 2018},
  pages 940--951, 2018.

\bibitem[DKK{\etalchar{+}}18b]{DKKMS18a}
Irit Dinur, Subhash Khot, Guy Kindler, Dor Minzer, and Muli Safra.
\newblock Towards a proof of the 2-to-1 games conjecture?
\newblock In {\em Proceedings of the 50th Annual {ACM} {SIGACT} Symposium on
  Theory of Computing, {STOC} 2018, Los Angeles, CA, USA, June 25-29, 2018},
  pages 376--389, 2018.

\bibitem[dlVKKR03]{VegaKKR03}
Wenceslas~Fernandez de~la Vega, Marek Karpinski, Claire Kenyon, and Yuval
  Rabani.
\newblock Approximation schemes for clustering problems.
\newblock In {\em Proceedings of the 35th Annual {ACM} Symposium on Theory of
  Computing, June 9-11, 2003, San Diego, CA, {USA}}, pages 50--58, 2003.

\bibitem[DS14]{dinur2014analytical}
Irit Dinur and David Steurer.
\newblock Analytical approach to parallel repetition.
\newblock In {\em Proceedings of the forty-sixth annual ACM symposium on Theory
  of computing}, pages 624--633. ACM, 2014.

\bibitem[Fei98]{F98}
Uriel Feige.
\newblock A threshold of ln \emph{n} for approximating set cover.
\newblock {\em J. {ACM}}, 45(4):634--652, 1998.

\bibitem[FL11]{feldman2011unified}
Dan Feldman and Michael Langberg.
\newblock A unified framework for approximating and clustering data.
\newblock In {\em Proceedings of the forty-third annual ACM symposium on Theory
  of computing}, pages 569--578. ACM, 2011.

\bibitem[GBH98]{guttmann1998approximation}
Nili Guttmann-Beck and Refael Hassin.
\newblock Approximation algorithms for min-sum p-clustering.
\newblock {\em Discrete Applied Mathematics}, 89(1-3):125--142, 1998.

\bibitem[GI03]{GI03}
Venkatesan Guruswami and Piotr Indyk.
\newblock Embeddings and non-approximability of geometric problems.
\newblock In {\em Proceedings of the Fourteenth Annual {ACM-SIAM} Symposium on
  Discrete Algorithms, January 12-14, 2003, Baltimore, Maryland, {USA.}}, pages
  537--538, 2003.

\bibitem[GK99]{GuK99}
Sudipto Guha and Samir Khuller.
\newblock Greedy strikes back: Improved facility location algorithms.
\newblock {\em J. Algorithms}, 31(1):228--248, 1999.

\bibitem[GL15]{guruswami2015inapproximability}
Venkatesan Guruswami and Euiwoong Lee.
\newblock Inapproximability of h-transversal/packing.
\newblock In {\em Approximation, Randomization, and Combinatorial Optimization.
  Algorithms and Techniques (APPROX/RANDOM 2015)}. Schloss
  Dagstuhl-Leibniz-Zentrum fuer Informatik, 2015.

\bibitem[GT17]{ghosh2017weak}
Mrinalkanti Ghosh and Madhur Tulsiani.
\newblock From weak to strong lp gaps for all csps.
\newblock In {\em 32nd Computational Complexity Conference (CCC 2017)}. Schloss
  Dagstuhl-Leibniz-Zentrum fuer Informatik, 2017.

\bibitem[Ind00]{indyk2000high}
Piotr Indyk.
\newblock {\em High-dimensional computational geometry}.
\newblock PhD thesis, Citeseer, 2000.

\bibitem[KLM19]{KLM18}
{Karthik {C. S.}}, Bundit Laekhanukit, and Pasin Manurangsi.
\newblock On the parameterized complexity of approximating dominating set.
\newblock {\em J. {ACM}}, 66(5):33:1--33:38, 2019.

\bibitem[KMN{\etalchar{+}}02]{kanungo2002local}
Tapas Kanungo, David~M Mount, Nathan~S Netanyahu, Christine~D Piatko, Ruth
  Silverman, and Angela~Y Wu.
\newblock A local search approximation algorithm for k-means clustering.
\newblock In {\em Proceedings of the eighteenth annual symposium on
  Computational geometry}, pages 10--18, 2002.

\bibitem[KMS17]{KMS17}
Subhash Khot, Dor Minzer, and Muli Safra.
\newblock On independent sets, 2-to-2 games, and grassmann graphs.
\newblock In {\em Proceedings of the 49th Annual {ACM} {SIGACT} Symposium on
  Theory of Computing, {STOC} 2017, Montreal, QC, Canada, June 19-23, 2017},
  pages 576--589, 2017.

\bibitem[KMS18]{KMS18}
Subhash Khot, Dor Minzer, and Muli Safra.
\newblock Pseudorandom sets in grassmann graph have near-perfect expansion.
\newblock In {\em 59th {IEEE} Annual Symposium on Foundations of Computer
  Science, {FOCS} 2018, Paris, France, October 7-9, 2018}, pages 592--601,
  2018.

\bibitem[KS12]{KS12}
Subhash Khot and Rishi Saket.
\newblock Hardness of finding independent sets in almost q-colorable graphs.
\newblock In {\em 53rd Annual {IEEE} Symposium on Foundations of Computer
  Science, {FOCS} 2012, New Brunswick, NJ, USA, October 20-23, 2012}, pages
  380--389, 2012.

\bibitem[LSW17]{LSW17}
Euiwoong Lee, Melanie Schmidt, and John Wright.
\newblock Improved and simplified inapproximability for k-means.
\newblock {\em Inf. Process. Lett.}, 120:40--43, 2017.

\bibitem[Man20]{M20}
Pasin Manurangsi.
\newblock Tight running time lower bounds for strong inapproximability of
  maximum k-coverage, unique set cover and related problems (via t-wise
  agreement testing theorem).
\newblock In {\em SODA}, 2020.

\bibitem[Mat00]{Matousek00}
Ji{\v{r}}{\'{i}} Matou{\v{s}}ek.
\newblock On approximate geometric k-clustering.
\newblock {\em Discrete {\&} Computational Geometry}, 24(1):61--84, 2000.

\bibitem[Sch00]{schulman2000clustering}
Leonard~J Schulman.
\newblock Clustering for edge-cost minimization.
\newblock In {\em STOC}, volume~5, 2000.

\bibitem[SG76]{sahni1976p}
Sartaj Sahni and Teofilo Gonzalez.
\newblock P-complete approximation problems.
\newblock {\em Journal of the ACM (JACM)}, 23(3):555--565, 1976.

\end{thebibliography}

\appendix

\section{Inapproximability of Continuous \kmean and \kmed in $\ell_\infty$-metric with Centers from Integral Lattice}\label{sec:lattice}

\begin{theorem}[\kmean with centers from integral lattice in $n^{O(1)}$ dimensional $\ell_\infty$-metric space]\label{thm:kmeanellinflattice} For any constant $\varepsilon>0$,
given a point-set $\Po\subset \R^m$ of size $n$ (and $m=\poly(n)$) and a parameter $k$ as input, it is \NP-hard to distinguish between the following two cases:
\begin{itemize}
\item \textbf{\emph{Completeness}}:  
There exists $\C':=\{c_1,\ldots ,c_k\}\subseteq \mathbb Z^m$ and $\sigma:\Po\to\C'$ such that $$\sum_{a\in\Po}\left(\|a-\sigma(a)\|_\infty\right)^2\le n,$$
\item \textbf{\emph{Soundness}}: For every $\C':=\{c_1,\ldots ,c_k\}\subseteq \mathbb Z^m$ and every $\sigma:\Po\to\C'$ we have: $$\sum_{a\in\Po}\left(\|a-\sigma(a)\|_\infty\right)^2\ge (9-\varepsilon)\cdot n.$$
\end{itemize}
\end{theorem}

\begin{theorem}[\kmed with centers from integral lattice in $n^{O(1)}$ dimensional $\ell_\infty$-metric space]\label{thm:kmedellinflattice} For any constant $\varepsilon>0$,
given a point-set $\Po\subset \R^m$ of size $n$ (and $m=\poly(n)$) and a parameter $k$ as input, it is \NP-hard to distinguish between the following two cases:
\begin{itemize}
\item \textbf{\emph{Completeness}}:  
There exists $\C':=\{c_1,\ldots ,c_k\}\subseteq \mathbb Z^m$ and $\sigma:\Po\to\C'$ such that $$\sum_{a\in\Po}\|a-\sigma(a)\|_\infty\le n,$$
\item \textbf{\emph{Soundness}}: For every $\C':=\{c_1,\ldots ,c_k\}\subseteq \mathbb Z^m$ and every $\sigma:\Po\to\C'$ we have: $$\sum_{a\in\Po}\|a-\sigma(a)\|_\infty\ge (3-\varepsilon)\cdot n.$$
\end{itemize}
\end{theorem}

\begin{proof}[Proof of Theorems~\ref{thm:kmeanellinflattice} and \ref{thm:kmedellinflattice}]
The proof follows as with the proof of Theorems~\ref{thm:kmeanellinfhighdim} and \ref{thm:kmedellinfhighdim}, but we have the following construction of $A:V\to\mathbb{R}^{m}$.
For every $v\in V$ and every $(u',v')\in E$, we define the $(u',v')^{\text{th}}$ coordinate of $A(v)$ as follows
\begin{align*}
A(v)_{(u',v')}:=\begin{cases}
1.5\text{\ \ \ \ \ \  if }v=u'\\
-0.5\text{\ \ \ if }v=v'\\
0.5\text{\ \ \ \ \ \  otherwise}
\end{cases}.
\end{align*}
\end{proof}

\end{document}